\documentclass[12pt,reqno]{amsart}
\usepackage{amsmath}
\usepackage{amsgen,amstext,amsbsy,amsopn,amsthm,amssymb,amsxtra}
\usepackage{graphicx}

\theoremstyle{definition}

\newcommand{\beq}{\begin{equation}}
\newcommand{\eeq}{\end{equation}}

\newtheorem{theorem}{Theorem}[section]

\newcommand{\version}{July 30, 2012}

\newtheorem{Lemma}{Lemma}[section]
 \newtheorem{Theorem}{THEOREM}[section]
\theoremstyle{definition}

\usepackage{color}

\newcommand{\E}{\mathcal{E}}
\renewcommand\epsilon\varepsilon

\theoremstyle{definition}

\newcommand{\R}{\mathbb{R}}

\begin{document}

\title[Disordered Bose Einstein Condensates]{{\bf Disordered Bose Einstein Condensates\\ with Interaction in One Dimension}}

\author[R. Seiringer]{Robert Seiringer}
 \address{Department of Mathematics and Statistics, McGill University, 805 Sherbrooke Street West, Montreal, QC H3A 2K6, Canada}
\email{robert.seiringer@mcgill.ca}

\author[J. Yngvason]{Jakob Yngvason}
 \address{Fakult\"at f\"ur Physik, Universit{\"a}t Wien,	\\ \normalsize\it Boltzmanngasse 5, 1090 Vienna, Austria}
\email{jakob.yngvason@univie.ac.at}

\author[V. A. Zagrebnov]{Valentin A.Zagrebnov}\address{D\'{e}partement de Math\'{e}matiques,
Universit\'e d'Aix-Marseille (AMU) and Centre de Physique
Th\'eorique - UMR 7332,  Luminy Case 907, 13288 Marseille, Cedex09, France}
\email{valentin.zagrebnov@univ-amu.fr}

\date{
\version}
\newcommand\const{{\rm const\,}}
\newcommand{\eps}{\varepsilon}

\begin{abstract}
We study the effects of random scatterers on the ground state of the one-dimensional Lieb-Liniger model of interacting 
bosons on the unit interval in the Gross-Pitaevskii regime. We prove that Bose Einstein condensation survives even a 
strong random potential with a high density of scatterers. The character of the wave function of the condensate, however, 
depends in an essential way on the interplay between randomness and the strength of the two-body interaction. For low 
density of scatterers or strong interactions the wave function extends over the whole interval. High density of scatterers 
and weak interaction, on the 
other hand,  leads to localization of the wave function in a fragmented subset of the interval. 
\end{abstract}

\maketitle

\section{Introduction}

While the effects of random potentials on single particle Schr\"odinger operators \cite{pasturfigotin} and ideal Bose gases 
\cite{lenoblepasturzagrebnov}--\cite{JPZ} are rather 
well explored, the present understanding of such effects on many-body systems of interacting particles is much less complete.
In recent years, however,  many papers concerning the interplay of Bose-Einstein Condensation (BEC) and disorder have 
appeared of which references \cite{gimperlein}--\cite{cardosoetal} are but a sample.

In this paper we present results on a model that is in a sense the
simplest one imaginable where this interplay can be studied by
rigorous mathematical means. This is the one-dimensional Lieb-Liniger
model \cite{LL} of bosons with contact interaction in a `flat' trap,
augmented by an external random potential that is generated by Poisson
distributed point scatterers of equal strength.  We study the ground
state and prove that, no matter the strength of the random potential,
BEC is not destroyed by the random potential in the Gross-Pitaevskii
(GP) limit where the particle number tends to infinity while the
coupling parameter in a mean-field scaling stays fixed.  The character
of the wave function of the condensate, however, depends in an
essential way on the relative size of the three parameters
involved. These are the scaled coupling parameter $\gamma$ for the
interaction among the particles, the density of the scatterers $\nu$,
and the strength $\sigma$ of the scattering potential. All these
parameters are assumed to be large in suitable units.

The investigation consists of three parts. In the following Section 2
we prove Bose-Einstein condensation in the ground state of our model
in a GP limit. This result holds, in fact, for quite general external
potentials, provided they are bounded below. Specific properties of
the random potential enter only in the estimate of the energy gap of
an effective mean-field Hamiltonian that sets a bound to the depletion
of the condensate for finite particle number.

Section 3 is concerned with the properties of the condensate in
dependence of the parameters.  In particular, for $\sigma\gg 1$ we
identify three different ``phases" of the condensate. For large
$\gamma\gg \nu^2$ the condensate is extended over the whole trap (the
unit interval in our model). A transition from a delocalized to a
localized state takes place when $\gamma$ is of the order $\nu^2$, in
the sense that for $\gamma\ll \nu^2$ the density is essentially
distributed among a fraction $\lambda \ll 1$ of the $\nu\gg 1$ intervals
between the obstacles. For $\gamma\gg\nu/(\ln \nu)^2$ we still have
$\lambda\nu\gg 1$, but for $\gamma\sim\nu/(\ln \nu)^2$ the fraction of
intervals that are significant occupied shrinks to $O(\nu^{-1})$. We
stress, however, that in all cases there is complete BEC into a single
state in the limit when the particle number tends to infinity.

In an appendix, we prove an estimate on the energy gap of
one-dimensional Schr\"odinger Hamiltonians. This estimate, applied to
a mean field Hamiltonian, is needed to establish a lower
bound for the condensate fraction that depends only on the parameters
of the problem and not on the individual realizations of the random
potential. The estimate on the energy gap is a generalization of a
result from \cite{kirschsimon} and is of independent interest.
Finally, for convenience of the reader, we collect in a second
appendix some facts about the Poisson distribution of obstacles.

Our study of the Gross-Pitaevskii energy functional in Section 3 can be regarded as complementary to the investigations in 
\cite{gimperlein, luganetal1, sanchesetal1, pikovskietal, luganetal, aleiner, piraudetal} where transitions between 
delocalization and localization for solutions of the GP equation (also the time-dependent one) are considered from 
different points of view. Among these papers \cite{luganetal1} is the one 
closest to ours as far as the questions asked are concerned, and it is appropriate to make a brief comparison.

One difference is that the random potentials considered in \cite{luganetal1} are intended to model experimentally generated 
laser speckles and somewhat different from the present model. In both cases, however, there is a length scale associated 
with the random potential ($\nu^{-1}$ in our case, $\sigma_R$ in \cite{luganetal1}), and the strength of the random potential 
that is denoted by $V_R$ in \cite{luganetal} is the analogue of our $\nu\sigma$. 
Moreover, our coupling parameter $\gamma$ may be compared with the chemical potential $\mu$ in \cite{luganetal1}.

The main difference, however,  is that in \cite{luganetal1} BEC is stated to hold only for large $\mu$ (analogous to large 
$\gamma$ in our model), while for weaker coupling a breakdown of BEC through a transition to a fragmented condensate and 
finally to a Lifschitz glass phase is claimed. We, however, {\it prove rigorously} that complete BEC holds in the whole  
parameter range considered when the particle number tends to infinity. Some further comments on this issue will be made at 
the end of Section 2.

\section{The Many-Body Model and BEC}

The model we consider is the Lieb-Liniger model of bosons with contact
interaction on the unit interval, with an additional external
potential $V$.  In the next section we shall take $V$ to be a sum of
delta functions at random points in the unit interval, but for the
proof of the energy bounds and BEC in the present section the only
assumption used is that $V$ is a {\it nonnegative} sum of delta
functions and a locally integrable function.  The Hamiltonian on the
Hilbert space $L^2([0,1], dz)^{\otimes_{\rm symm}^N}$ is
\beq\label{ham}
H:=\sum_{i=1}^N\left(-\partial_{z_i}^2+V(z_i)\right)+\frac\gamma
N\sum_{i<j}\delta(z_i-z_j)\eeq with {$\gamma\geq 0$}, and we assume
Dirichlet boundary conditions at the end points of the unit
interval. The Hamiltonian has a unique ground state and we denote the
ground state energy by $E^{\rm QM}_0$.

The Gross-Pitaevskii energy functional is defined by
\begin{equation}\label{gpfunc}\mathcal E^{\rm GP}[\psi]:=\int_0^1\left\{|\psi'(z)|^2+V(z)|\psi(z)|^2+(\gamma/2)|\psi(z)|^4\right\}dz. \end{equation}
Again we assume Dirichlet boundary conditions. By standard methods one proves that there is a unique nonnegative minimizer 
$\psi_0\in H_0^1([0,1])$, normalized so that $\int \psi_0^2(z)dz=1$,  
with a corresponding energy $e^{\rm GP}$. This is equal to the ground state energy $e_0$ of the mean-field Hamiltonian
\beq \label{mfh}h:=-\partial_z^2+V(z)+\gamma\psi_0(z)^2-\frac \gamma 2\int\psi_0^4\eeq
with $\psi_0$ as the corresponding eigenfunction.

\begin{Theorem}[Energy bounds]
\beq\label{gse} e_0\geq \frac{E^{\rm QM}_0}N\geq e_0\left(1-({\rm const.}) N^{-1/3}\min\{\gamma^{1/2},\gamma\}\right).\eeq
\end{Theorem}

Note that the error term  depends only on $\gamma$ besides $N$ and is independent of $V$ (as long as  $V \geq 0$).
\smallskip

We denote the second eigenvalue of $h$ by $e_1$ and by $N_0$ the average occupancy of $\psi_0$ in the many-body ground state $\Psi_0$ of $H$, i.e.,
\beq N_0:=
{\rm tr}\, \rho^{(1)} |\psi_0\rangle\langle\psi_0|\eeq
where  $\rho^{(1)}$ is the one-particle reduced density matrix of $\Psi_0$. The depletion of the condensate is $\left(1-\frac {N_0}N\right)$.

\begin{Theorem}[BEC]\label{thm2}
\beq\label{depletion} \left(1-\frac {N_0}N\right)\leq ({\rm const.}) \frac {e_0}{e_1-e_0}
N^{-1/3}\min\{\gamma^{1/2},\gamma\}.
\eeq
\end{Theorem}
\smallskip 

Theorem~\ref{thm2} implies in particular complete BEC in the limit
$N\to\infty$ with $\gamma$ and $V$ fixed. In fact, Eq.\ \eqref{depletion}
implies that the reduced one-particle density matrix divided by $N$
converges in trace norm to the projector on $\psi_0$.  The energy gap
$e_1-e_0$ is always strictly positive since $h$ has a unique ground
state. Its dependence on $\gamma$ and $V$ is discussed in the
Appendix. For the random external potentials introduced in
\eqref{randompot} below, Eq.\ \eqref{gap1} implies that complete BEC
holds not only for almost every random sample of scatterers, but the
depletion of the condensate goes to zero even in $L^p$ norm on the
probability space for every $p<\infty$.

Our proof of BEC in the GP limit is simpler than the corresponding
proof in three and two dimensions \cite{LS} because the one-dimensional case considered
here corresponds to a high-density, mean-field limit. In contrast, the
work in \cite{LS} deals with a dilute limit that requires quite
different tools.

\medskip

\subsection{Proof of Theorems 2.1 and 2.2}

The upper bound in \eqref{gse} is immediately obtained by using $\psi_0^{\otimes N}$ as a trial function for the expectation 
value of the Hamiltonian \eqref{ham}.

For the proof of the lower bound in \eqref{gse}, as well as the proof of  \eqref{depletion}, we begin by splitting off part of the kinetic energy 
and write, with $\eps>0$ and $p_i:=-{\mathrm i}\partial_{z_i}$,
\beq
H=\sum_{i=1}^N\left\{\left(1-\frac{N-1}{2N} \eps\right)p_i^2+V(z_i)\right\}+\frac 1N\sum_{i<j}\left[\frac \eps 2(p_i^2+p_j^2)+
\gamma\delta(z_i-z_j)\right].
\eeq
Next we estimate the term in square brackets  from below by a regular interaction potential using the bound
\beq \label{wbound}-\eps \partial_z^2+\gamma\delta(z)\geq w(z)\eeq
with
\beq w(z):=\frac {\gamma}{1+(b\gamma/2\eps)}\,\frac 1{2b}\,{\exp}(-|z|/b)
\eeq
for any $b>0$, cf. \cite[Lemma~6.3]{BSY}. Note that
\beq \delta_b(z):=\frac 1{2b}\,{\exp}(-|z|/b)\to \delta(z)\eeq
as $b\to 0$, and $\int \delta_b(z)dz=1$ for all $b$.
From \eqref{wbound} we obtain the following inequality (in the sense of quadratic forms on wave functions that satisfy Dirichlet 
boundary conditions at the boundary of $[0,1]^{ N}$ and are identically zero outside of $[0,1]^{ N}$)
\beq
H\geq\sum_{i=1}^N\left\{\left(1-\frac\eps 2\right)p_i^2+V(z_i)\right\}+\frac {\tilde\gamma}N\sum_{i<j}\delta_b(z_i-z_j)
\eeq
where we have denoted $\gamma/[1+(b\gamma/2\eps)]$ by $\tilde\gamma$ for short, and bounded $-(N-1)\eps/2N$ from below by 
$-\eps/2$. 
Since $\delta_b$ is of positive type  we can, for any density $\rho(z)$, bound the last term by a sum of one-body terms:
\beq\label{deltabbound}
\frac {\tilde\gamma}N\sum_{i<j}\delta_b(z_i-z_j)\geq \sum_{i=1}^N\left\{\tilde\gamma\,\delta_b*\rho(z_i)-
\frac {\gamma}2\int(\delta_b*\rho)\rho-\frac \gamma{2N}\delta_b(0)\right\}
\eeq
where we have used that $\gamma\geq\tilde\gamma$. We take $\rho$ to be the GP density $\psi_0^2$ (and zero outside of $[0,1]$).
For the comparison of the mean field Hamiltonian \eqref{mfh} with the right side of \eqref{deltabbound} we use the following 
simple estimates (recall that all integrals are effectively over the unit interval):
\beq\Vert \delta_b*\rho\Vert_\infty\leq\Vert\psi_0\Vert^2_\infty\leq \Vert \psi_0\Vert_2 \Vert \psi_0'\Vert_2\leq  e_0^{1/2}
\eeq
and
\begin{multline}\label{simplebound2}\Vert \delta_b*\rho-\rho\Vert_\infty=\sup_x\left|\int[\rho(x-y)-\rho(x)]\delta_b(y)dy\right|
\leq\sup_x\left|\int\int_{x-y}^x|\rho'(u)|du\, \delta_b(y)dy\right|\\ 
\leq 2\Vert\psi_0\Vert_\infty\, e_0^{1/2}\int|y|^{1/2}\delta_b(y)dy\leq 2^{1/2}\pi^{1/2}\,e_0^{3/4}\, b^{1/2}.
\end{multline}
Using $\tilde\gamma\geq \gamma-b\gamma^2/(2\eps)$, Eqs.\ \eqref{deltabbound}--\eqref{simplebound2} and \eqref{mfh} lead to
\beq\label{basicbound} H\geq \sum_{i=1}^N \left\{h^{(i)}-\frac\eps 2 p_i^2-\frac {b\gamma^2}{2\eps}e_0^{1/2}-
\frac{3\pi^{1/2}}{2^{1/2}}\gamma\, e_0^{3/4} b^{1/2}-\frac\gamma{4Nb}
\right\}.\eeq
In the second term we estimate $p_i^2$ by $e_0$ leading to an optimal $\eps\sim b^{1/2}\gamma e_0^{-1/4}$. Optimizing the resulting expression
\beq ({\rm const.})\gamma\left[b^{1/2} e_0^{3/4}+ (Nb)^{-1}\right]\eeq
over $b$ gives $b\sim N^{-2/3}e_0^{-1/2}$ and
\beq \label{gammalargeest}\frac{E_0^{\rm QM}}N\geq e_0\left(1-({\rm const.})N^{-1/3}\gamma e_0^{-1/2}\right).
 \eeq
For large $\gamma$ we may use  the bound $e_0\geq (\gamma/2)\int \rho^2\geq (\gamma/2)$
while for $\gamma \lesssim1$ we use that  $e_0\geq \pi ^2$ since $\pi ^2$ is the lowest eigenvalue of $-\partial_z^2$ on the 
unit interval with Dirichlet boundary conditions. Altogether we obtain  \eqref{gse}.

To prove Theorem~\ref{thm2} we trace the one-particle operator on the right side of \eqref{basicbound} with the one-particle reduced 
density matrix of the ground state  of $H$ and obtain the more precise lower bound
\beq \frac{E_0^{\rm QM}}N\geq \frac{N_0}N e_0+\left(1-\frac{N_0}N\right)e_1-({\rm const.})\,e_0\, N^{-1/3}\min\{\gamma^{1/2},\gamma\}\eeq
where $e_1$ is the energy of the first excited state of the mean field Hamiltonian $h$. Combined with the upper bound
${E_0^{\rm QM}}\leq N e_0$ this leads to the claimed estimate \eqref{depletion} for the depletion of the condensate.\medskip

\subsection{Remarks}
1. If $N_{\leq k}$ denotes the occupation of the $k$ lowest
eigenvalues of $h$ with energies up to $e_k$, then the following
generalization of \eqref{depletion} holds: \beq\label{depletion2}
\left(1-\frac {N_{\leq k}}N\right)\leq C \frac {e_0}{e_k-e_0}
N^{-1/3}\min\{\gamma^{1/2},\gamma\}.  \eeq The proof is the same
way as for \eqref{depletion}, using again the bound
\eqref{basicbound}. For {\it finite} $N$ the bounds \eqref{depletion}
and \eqref{depletion2} are, of course, only useful if their respective
right-hand sides are less than 1. This may hold for \eqref{depletion2}
even if it does not for \eqref{depletion}.

2. As remarked in the Introduction, the authors of \cite{luganetal1}
claim a transition from a BEC to a fragmented condensate and finally
to a Lifschitz glass phase in their model. Since no rigorous limit of
large particle numbers is considered in \cite{luganetal1} there is no
contradiction with the complete BEC that is proved in the present
paper in the whole parameter range. In fact, for a given {\it finite}
particle number $N$ a fragmented condensate, i.e., $(1-N_{\leq k}/N)$
small for some $1\ll k\ll N$, may be a reasonable substitute for the
fully condensed state that emerges according to \eqref{depletion} in
the large $N$ limit. We also point out that as far as the density in
the position variable is concerned, a fragmented condensate with
non-overlapping single-particle wave functions is indistinguishable
from a fully condensed state where the wave function of the
condensate is a coherent superposition of the non-overlapping
functions. The difference shows up, however, in the momentum
distribution.

\section{The Gross-Pitaevskii Theory}

Having established the GP minimizer as the wave function of the
condensate in the Gross-Pitaevskii limit of the many-body problem we
now turn to the study of the dependence of this minimizer on a random
external potential. Specifically, we shall take $V$ in \eqref{gpfunc}
to be \beq\label{randompot} V(z)=\sigma V_\omega(z)\eeq with
$\sigma\geq 0$ and
\begin{equation}\label{defvo}
V_\omega(z) := \sum_i \delta(z-z_i/\nu)
\end{equation}
where $\nu>0$ and $\omega$ is a variable in a probability space
generating the random points $z_i$ in $\mathbb{R}$ which are assumed
to be Poisson distributed with density $\lambda=1$ (see Appendix 2).
The distances $\ell_i$ between neighboring points $z_i/\nu$ and
$z_{i+1} /\nu$ are then independent random variables distributed
according  to the exponential law (cf.\ Eq.\eqref{marginals})
\begin{equation}\label{exp:law}
dp_\nu(\ell) =  \nu e^{-\ell\nu} d\ell\,.
\end{equation}
Thus, with probability one only finitely many of these points are in $[0,1]$, and on average there are $\nu$ such points. 
We start by considering the energy in an interval of length $\ell$ between two of the points, which after a translation 
and scaling
\beq
z\to x:=z/\ell, \quad \sigma\to \alpha:=\ell \sigma  , \quad \gamma\to \kappa:=\ell\gamma
\eeq
can be conveniently taken to be the unit interval.

\subsection{An Auxiliary Problem}\label{ss:aux}

For $\kappa\geq 0$ and $\alpha\geq 0$, let $e(\kappa,\alpha)$ denote the auxiliary GP energy
\begin{equation}\label{GP-energ}
e(\kappa,\alpha) := \inf_{\|\phi\|_{2} =1} \E_{\kappa,\alpha}[\phi] \ ,
\end{equation}
where
\begin{equation}
\E_{\kappa,\alpha}[\phi]:= \int_0^1 dx \ \left( |\phi'(x)|^2 + \frac \kappa 2 |\phi(x)|^4\right) +
\frac \alpha 2 \left( |\phi(0)|^2 + |\phi(1)|^2 \right)  \label{GP-funct}
\end{equation}
for $\phi \in H^1([0,1])$. Standard methods show that there exists a minimizer for (\ref{GP-energ}), which we denote by
${\phi}_{\kappa,\alpha}$, i.e.,
\begin{equation}\label{funct-min}
e(\kappa,\alpha) = \E_{\kappa,\alpha}[{\phi}_{\kappa,\alpha}] \ .
\end{equation}
The minimizer is unique up to a constant phase factor, which can be chosen such that $\phi_{\kappa,\alpha}$ is non-negative.

Note that, as an infimum over linear functions, $e(\kappa,\alpha)$ is jointly concave in $\kappa$ and $\alpha$.
For us it will also be important that $\kappa\mapsto \kappa\, e(\kappa,\alpha)$ is strictly convex. This follows from the fact that
\begin{equation}
\kappa\, e(\kappa,\alpha) = \inf \left\{ \E_{1,\alpha}(|\phi|) \, : \, \int_0^1 |\phi|^2 = \kappa\right\}
\end{equation}
together with the strict convexity of $|\phi|^2 \mapsto \E_{1,\alpha}(|\phi|)$.

We have the simple bounds
\begin{equation}\label{Estim-E}
e(\kappa,\alpha) = \E_{0,\alpha}(\phi_{\kappa,\alpha}) + \frac \kappa 2 \int_0^1 dx \ |\phi_{\kappa,\alpha}(x)|^4 \geq  e(0,\alpha) + \frac \kappa 2
\end{equation}
and
\begin{equation}\label{Estim-E2}
e(\kappa,\alpha) \leq \E_{\kappa,\alpha}(\phi_{0,\alpha}) = e(0,\alpha) + \frac \kappa 2  \int_0^1 dx \ |\phi_{0,\alpha}(x)|^4\ ,
\end{equation}
where we used the Schwarz inequality to obtain the first inequality.

Obviously $e(0,0)=0$ and $e(0,\infty)=\pi^2$. The minimizer
$\phi_{0,\alpha}$ for $\kappa=0$ is of the form $\phi_{0,\alpha}(x)=
a_{\alpha} \cos [b_{\alpha} (x -1/2)]$, where $b_\alpha\geq 0$
increases from $0$ to $\pi$ as $\alpha$ increases from $0$ to
$\infty$, and $a_{\alpha}>0$ is determined by the normalization
condition $\|\phi_{0,\alpha}\|_2 =1$. A close inspection shows that
$\|\phi_{0,\alpha}\|_4$ is monotone increasing in $\alpha$, and hence
\begin{equation}\label{E3}
\int_0^1 dx \ |\phi_{0,\alpha}(x)|^4 \leq \int_0^1 dx \ |\phi_{0,\infty}(x)|^4 = \frac 32\ .
\end{equation}
In particular, from (\ref{Estim-E})--(\ref{E3}) we obtain
\begin{equation}\label{32c}
\frac 12 \leq \frac{e(\kappa,\alpha)-e(0,\alpha)}{\kappa}\leq \frac  34
\end{equation}
for all $\kappa>0$ and $\alpha\geq 0$.

\begin{Lemma}\label{lem1}
{\it For some constant $C$ independent of $\kappa$ and $\alpha$ we have
\begin{equation}\label{tos}
e(\kappa,\infty) \geq e(\kappa,\alpha) \geq e(\kappa,\infty) \left( 1 - C \alpha^{-1/2} \right) \ .
\end{equation}}
\end{Lemma}

\begin{proof}
  The first inequality follows clearly from monotonicity of
  $e(\kappa,\alpha)$ in $\alpha$.  To prove the second, it is enough
  to consider $\alpha > C^2$ (for otherwise the right side of
  (\ref{tos}) is negative). Using a piecewise linear function which is
  constant for $x\in(\epsilon,1-\epsilon)$ as a trial function, we obtain (after optimizing over $\epsilon$) that
\begin{equation}\label{est1}
e(\kappa,\infty) \leq   \frac \kappa 2 \left( 1  + C  \kappa^{-1/2} \right)
\end{equation}
for some $C>0$ and $\kappa$ large enough. In particular, since obviously $e(\kappa,\alpha) \geq \kappa/2$, the second inequality 
in (\ref{tos}) follows right away if $\kappa > \alpha$.

We are left with the case $\alpha \geq \kappa$. Since
\begin{equation}
e(\kappa,\infty) \geq e(\kappa,\alpha) \geq \frac \kappa 2 + \alpha \phi_{\kappa,\alpha}(0)^2
\end{equation}
(using the symmetry of $\phi_{\kappa,\alpha}$ with respect to reflections at $1/2$), we conclude from (\ref{est1}) that
\begin{equation}
 \phi_{\kappa,\alpha}(0)^2 \leq C \frac{\kappa^{1/2}}\alpha \leq C \alpha^{-1/2}
\end{equation}
in this case. To obtain the desired bound, we use use $\lambda(
\phi_{\kappa,\alpha}(x) - \phi_{\kappa,\alpha}(0))$ as a trial
function for $e(\kappa,\infty)$, where $\lambda$ is an appropriate
normalization factor. A simple calculation then yields (\ref{tos}).
\end{proof}

For later use we shall also introduce the Legendre transform of the map $\kappa \mapsto \kappa\, e(\kappa, \alpha)$,
\begin{equation}\label{GCE-one-box}
g(\mu,\alpha):= \inf_{n\geq 0} \left(n\, e(n, \alpha) - \mu n \right)\,.
\end{equation}
From the concavity of $\kappa\mapsto e(\kappa,\alpha)$ and the
uniqueness of the minimizer of (\ref{GP-funct}) one can easily
conclude differentiability of the map $\kappa\mapsto
e(\kappa,\alpha)$. Since $\kappa\mapsto \kappa\, e(\kappa,\alpha)$ is
strictly convex, the infimum in (\ref{GCE-one-box}) is uniquely attained. If it
is attained at some $n>0$, this $n$ satisfies the equation
\begin{equation}\label{inf-eq}
e(n , \alpha)  +  n \, e' (n , \alpha)  =
\mu  \ ,
\end{equation}
where we denote $e'(\kappa,\alpha):= \partial_{\kappa}e(\kappa,\alpha)$. We subtract the
energy $e(0, \alpha)$  on both sides of
(\ref{inf-eq}) for convenience. Consequently, we observe that the
infimum in (\ref{GCE-one-box}) is attained at $n$ satisfying
\begin{equation}\label{inf-n}
n = \left[\mu - e(0,\alpha)\right]_{+} \, \frac {1} {e'(n,\alpha) +
\tfrac{1}{n}(e(n,\alpha)-e(0,\alpha))} \ .
\end{equation}
Here $[t]_{+}$ stands for the positive part of a real number $t$. The
unique solution of (\ref{inf-n}) will be denoted by $\bar
n(\mu,\alpha)$.

From (\ref{32c}) and concavity of $e(\kappa,\alpha)$ is follows that
$1/2\leq e'(\kappa,\alpha)\leq 3/4$. Moreover, the second term in the
denominator in (\ref{inf-n}) lies in $[1/2,3/4]$, again by
(\ref{32c}). We thus conclude that the optimizing $n$ satisfies
\begin{equation}\label{bounds:n}
\frac 2{3}\left[\mu - e(0,\alpha)\right]_{+} \leq \bar n(\mu,\alpha)  \leq \left[\mu - e(0,\alpha)\right]_{+}\,.
\end{equation}
In particular, it is non-zero if and only if $\mu > e(0,\alpha)$, and we can write
\begin{equation}\label{inf-n-ord}
\bar n(\mu,\alpha) \sim \left[\mu - e(0,\alpha)\right]_{+} \ ,
\end{equation}
where $a\sim b$ means that $a/b$ is bounded from above and below by positive constants.

Finally, we note that since $e(0,\alpha)$ is increasing and concave in $\alpha$, there exists a constant $C>0$ such that
\begin{equation}\label{e0a}
e(0,\alpha) \geq \frac {C\alpha}{1+\alpha} \,.
\end{equation}
This bound will be useful later.


\subsubsection{The Average Energy}

We shall consider the average energy for a particle distributed over intervals of side length $\ell$, whose  distribution
is given by the exponential law \eqref{exp:law}
with $\nu >0$. More precisely, we define
\begin{equation}\label{altdef}
e_0(\gamma,\nu) = \inf \left\{ \nu \int_0^\infty dp_\nu(\ell) \,\frac{ n(\ell)}{\ell^2} e(n(\ell)\ell\gamma,\infty) \, : \, \nu \int_0^\infty dp_\nu(\ell) \, n(\ell) = 1 \right\}\,,
\end{equation}
where the infimum as over all $n:\mathbb{R}_+\to \mathbb{R}_+$ satisfying the normalization constraint.
With the aid of (\ref{GCE-one-box}), we can alternatively write
\begin{equation}\label{def:e0}
e_0(\gamma,\nu) = \sup_{\mu>0} \left\{ \mu + \nu \int_0^\infty dp_\nu(\ell)\, \frac 1{\ell^3 \gamma} g(\mu \ell^2,\infty) \right\}\,.
\end{equation}
From simple scaling it follows that
\begin{equation}
e_0(\gamma,\nu) = \gamma\, e_0(1, \nu/\sqrt{\gamma})\,.
\end{equation}
The infimum in (\ref{altdef}) is attained for $n(\ell) = (\ell\gamma)^{-1} \bar n(\mu \ell^2,\infty)$ for suitable $\mu>0$, 
being equal to the optimal $\mu$ in (\ref{def:e0}).
Using (\ref{inf-n-ord}) as well as the fact that
\begin{equation}
1\leq e^x \int_x^\infty e^{-t}\left(t-\frac {x^2}{t}\right) dt \leq 2
\end{equation}
we see that $\mu$ satisfies the relation
\begin{equation}\label{rel-1}
 1 \sim \frac\mu\gamma e^{-\pi\nu/\sqrt\mu} \ .
\end{equation}
In other words, $\mu \sim \gamma\, f(\nu^2/\gamma)$, where $f:\mathbb{R}_+\to \mathbb{R}_+$ denotes the function
\begin{equation}\label{def:f}
f(x) = \left\{ \begin{array}{ll} 1 & \text{for $x\leq 1$} \\ \frac x{ \left(1+ \ln x\right)^2} & \text{for $x\geq 1$.} 
\end{array} \right.
\end{equation}
Also $e_0(\gamma,\nu) \sim \gamma\, f(\nu^2/\gamma)$.


\subsection{The Gross-Pitaevskii Energy}

We now turn to the actual GP energy functional we want to consider,
\begin{equation}\label{GP-funct-ell}
\E^{\rm GP}_\omega[\psi]:= \int_{0}^1 \, dz \, \left( |\psi'(z)|^2 + \sigma V_\omega(z) |\psi(z)|^2+ \frac \gamma 2 |\psi(z)|^4\right)
\end{equation}
with $V_\omega$ as in \eqref{defvo}, and define the GP energy as
\begin{equation}\label{GP-energy}
e_\omega(\gamma,\sigma,\nu) := \inf_{\|\phi\|_{2} =1} \E^{\rm GP}_\omega[\phi] \ .
\end{equation}

\begin{Theorem}[Convergence of the energy]\label{thm:gp}
{\it Assume that $\nu\to \infty$, $\sigma\to \infty$ and $\gamma\to \infty$ in such a way that
\begin{equation}\label{asu}
\gamma\gg \frac{\nu}{\left(\ln \nu\right)^2} \quad \text{and} \quad
\sigma \gg \frac{ \nu}{1+ \ln \left( 1+ \nu^2/\gamma\right) }\,.
\end{equation}
Then, for almost every $\omega$,
\begin{equation}\label{thm:mainresult}
\lim \frac{ e_\omega(\gamma,\sigma,\nu)}{e_0(\gamma,\nu)} = 1
\end{equation}
where $e_0(\gamma,\nu)$ is the {\em deterministic} function defined in (\ref{def:e0}), which satisfies $e_0(\gamma,\nu) = 
\gamma \, e_0(1,\nu/\sqrt{\gamma}) \sim \gamma\, f(\nu^2/\gamma)$.}
\end{Theorem}

\subsubsection{Remark.} Without the interaction, i.e., for $\gamma=0$,
 the energy is not deterministic but is simply
the ground state energy in the random potential $\sigma V_\omega$. For
$\sigma\to\infty$ it becomes equal to the kinetic energy of
a particle localized in the largest subinterval free of obstacles and
hence, with probability 1, of the order $\nu^2/(\ln \nu)^2$, cf.\
Appendix 2.  \medskip

We shall now investigate the optimal particle number distribution in
(\ref{altdef}). Since $\bar n(\mu\ell^2,\infty)>0$ if and only
if $\mu \ell^2 > \pi^2$, the average number of intervals with non-zero
occupation numbers is given by
\begin{equation}\label{mean-box}
\nu \int_{\pi/\sqrt\mu}^\infty dp_\nu(\ell) = e^{-\pi\nu/\sqrt\mu} \, \nu =: \lambda\, \nu \ .
\end{equation}
Since $\nu$ is the total number of available intervals,
$\lambda \leq 1$ in (\ref{mean-box}) defines the \textit{fraction} of
them which are occupied.
With the help of the definition of $\lambda$ in (\ref{mean-box}) the relation (\ref{rel-1}) can also be written as
\begin{equation}\label{rel-2}
\frac{\pi^2\nu^2}{\left(\ln \lambda^{-1}\right)^2}= \mu \sim \frac\gamma\lambda  \ .
\end{equation}
In other words, $\lambda$ is determined by $\gamma$ and $\nu$ via the relation
\begin{equation}\label{rel-3}
\gamma \sim \frac{\lambda \ \nu^2}{\left(\ln\lambda^{-1}\right)^2} \ .
\end{equation}
We can distinguish the following {\bf limiting cases}:\medskip

\begin{itemize}
\item If $\gamma\gg \nu^2$ then by (\ref{rel-3}) we get $\lambda\to
  1$, i.e., {\em all} the intervals are occupied. The chemical potential
  satisfies $\mu \sim \gamma$ in this regime.\medskip
\item
If $\gamma\sim \nu^2$ than $\lambda \sim 1$, but $\lambda$ is strictly less than $1$. Again we have $\mu \sim \gamma$. \medskip
\item
If $\gamma\ll \nu^2$ then $\lambda \ll 1$, i.e., only a small fraction of the intervals are occupied. The relation (\ref{rel-2}) 
implies $\mu \sim (\nu/\ln(\nu^2/\gamma))^2$ for the chemical potential.\medskip
\item
If $\gamma\sim \nu/(\ln\nu)^2$ then by (\ref{rel-3}) the fraction $\lambda$ becomes $O(1/\nu)$, i.e.,
only \textit{finitely many} intervals are occupied. In this latter case, $\mu\sim \gamma\nu \sim \nu^2/(\ln\nu)^2$,
which corresponds exactly to the inverse of the square of the size of the \textit{largest} interval, cf.\ Appendix 2.
\medskip
\end{itemize}
\medskip

In particular, $\lambda \gg 1/\nu$ only if $\gamma \gg \nu/(\ln \nu)^2$, and hence this condition guarantees that many 
intervals are occupied. In this case the law of large numbers applies and hence the energy becomes deterministic in the 
limit. If $\lambda \nu = O(1)$, on the other hand, the value of $e_\omega(\gamma,\sigma,\nu)$ 
is random. This shows, in particular, that our condition (\ref{asu}) on $\gamma$ is optimal.

Also the second condition in (\ref{asu}) on $\sigma$ can be expected to be optimal. It can be rephrased as 
$\bar \ell \sigma \gg 1$, where $\bar \ell$ is the (weighted) average interval length
\begin{equation}
\bar \ell = \nu \int_0^\infty dp_\nu(\ell) \, \ell\, n(\ell)
\end{equation}
with $n(\ell) = (\ell\gamma)^{-1} \bar n(\mu \ell^2,\infty)$  the optimizer in (\ref{altdef}). A simple calculation shows 
that $\bar \ell \sim \nu^{-1} (1+\ln(1+\nu^2/\gamma))$.

\subsubsection{Remark} The conclusions above apply to the minimizing particle distribution in
(\ref{altdef}). In the parameter regime defined in (\ref{asu}), they
can be shown to apply also to the actual minimizer of the GP
functional (\ref{GP-funct-ell}). This follows in a standard way by
perturbing the energy functional, and then taking a derivative with
respect to the perturbation. Specifically, one can add a term $\beta
\int_\R f(|\psi(x)|) dx $ for suitable functions $f$ to the GP
functional, and change the definition of $e_0$ in (\ref{altdef})
accordingly. One then concludes that the energy asymptotics
(\ref{thm:mainresult}) continues to hold in this perturbed
case. Taking a derivative with respect to $\beta$ at $\beta=0$ yields
the desired information on the GP minimizer.

\bigskip

\subsection{Proof of Theorem~\ref{thm:gp}}
\subsubsection{Preliminaries}
Let $z_1\leq z_2\leq \dots \leq z_m$ denote those points in
(\ref{defvo}) which lie in the open interval $(0,\nu)$. Let also $z_0
= 0$ and $z_{m+1}=\nu$, and denote $\ell_i=|z_{i+1}-z_{i}|/\nu$ for $0\leq i
\leq m$.  Then $\sum_{j=0}^m \ell_j = 1$. For any $\psi\in
H_0^1([0,1])$, we have
\begin{equation}\label{splitting}
\E^{\rm GP}[\psi] = \sum_{j=0}^m  \frac{n_j}{\ell_j^2}  \E_{n_j \ell_j \gamma,\ell_j\sigma}[\psi_j]
\end{equation}
with
\begin{equation}
n_j = \int_{z_j}^{z_{j+1}} dx \,|\psi(x)|^2
\end{equation}
and
\begin{equation}
\psi_j (x) = \sqrt{\frac{\ell_j}{n_j}} \psi\left( z_j + \ell_j x \right) \quad \text{for $x\in [0,1]$.}
\end{equation}

We can choose $\psi\in H_0^1([0,1])$ such that $\psi_j \propto \phi_{\ell_j \gamma,\infty}$ for all $j$ and $\{n_j\}$ is an arbitrary collection of non-negative numbers adding up to $1$. In particular, we obtain the upper bound
\begin{equation}\label{upper}
e_\omega(\gamma,\sigma,\nu) \leq \inf_{\{n_j\}} \sum_{j=0}^m \frac{n_j}{\ell_j^2} e(n_j \ell_j \gamma,\infty)\ ,
\end{equation}
where the infimum is over all collections of non-negative numbers $\{n_0,n_1,\dots,n_m\}$ with $\sum_{j=0}^m n_j = 1$.

From (\ref{splitting}) we also have the lower bound
\begin{equation}\label{below}
e_\omega(\gamma,\sigma,\nu) \geq \inf_{\{n_j\}} \sum_{j=0}^m \frac{n_j}{\ell_j^{2}} e(n_j \ell_j \gamma,\ell_j \sigma)\ .
\end{equation}
From Lemma~\ref{lem1} we conclude that the two bounds (\ref{below})
and (\ref{upper}) agree, to leading order, as long as $\ell_j \sigma$
is large for all the intervals of length $\ell_j$. In particular, the
study of the GP energy reduces to a study of independent intervals in this
case. More precisely, it suffices that $\ell_j \sigma$ be large for those intervals containing most of the particles.


\subsubsection{Poisson Distribution}

In this subsection we shall collect a few useful facts about the Poisson distribution, we will be used later on. A more detailed 
discussion of some of the points   can be found in Appendix 2.
 Let $\chi_m(\omega)$ denote the characteristic function of the set corresponding to having exactly $m$ scatterers inside the 
 interval $(0,\nu)$. Its expectation value equals
\begin{equation}\label{ex:m}
\langle \chi_m \rangle = e^{-\nu} \frac{\nu^m}{m!} \,.
\end{equation}
We shall frequently use the bounds
\begin{equation}\label{p:u}
e^{-\nu} \sum_{m\geq \lambda\nu} \frac{\nu^m}{m!} \leq e^{-\nu} \sum_{m\geq 0} \frac{\nu^m}{m!} \lambda^{m-\lambda\nu} = 
e^{-\nu \left(1-\lambda + \lambda \ln \lambda\right)}   \quad \text{for $\lambda\geq 1$}
\end{equation}
and
\begin{equation}\label{p:l}
e^{-\nu} \sum_{m\leq \lambda\nu} \frac{\nu^m}{m!} \leq e^{-\nu} \sum_{m\geq 0} \frac{\nu^m}{m!} \lambda^{m-\lambda\nu} = e^{-\nu \left(1-\lambda + \lambda \ln \lambda\right)}   \quad \text{for $\lambda\leq 1$.}
\end{equation}
Note that $1-\lambda+\lambda\ln\lambda \geq 0$ for $\lambda\geq 0$, with equality only for $\lambda = 1$.

For any function depending only on $\ell_j$ for some fixed $1\leq j\leq m-1$, we have
\begin{equation}\label{dpl}
\left\langle f(\ell_j) \right \rangle = \int_0^\infty dp_\nu(\ell)\, f(\ell)
\end{equation}
with $dp_\nu$ defined in (\ref{exp:law}). This fact (cf.\ Eq.\ \eqref{marginals} in Appendix 2) will be used repeatedly below. Moreover, the variables  $\ell_j$ and $\ell_k$ for $k\neq j$ are independent, and hence $\langle f_1(\ell_j) f_2(\ell_k)\rangle = \langle f_1(\ell_j)\rangle\langle f_2(\ell_k)\rangle$ in this case.


\subsubsection{Upper Bound}\label{ss:ub}

Let $\mu=\mu(\gamma,\nu)$ be the optimizer in (\ref{def:e0}).
With $\bar n(\mu,\alpha)$ defined above, we then have
\begin{equation}\label{choice:mu}
 \nu \int_0^\infty dp_\nu(\ell)\, \frac 1{\ell \gamma} \bar n (\mu \ell^2,\infty) = 1\,.
\end{equation}
In every interval of length $\ell_j$, $1\leq j\leq m-1$,  we shall place $n_j = (\ell_j \gamma)^{-1} \bar n (\mu\ell_j^2,\infty)$ particles. For simplicity, we shall not place 
any particles in the first and last interval.  The total number of particles in the system is then $N = \sum_{j=1}^{m-1} n_j$.

\begin{Lemma}\label{lem:1}
If $\nu\to \infty$ and $\gamma\to\infty$ with $\gamma\gg \nu/(\ln \nu)^2$ then
\begin{equation}
\lim N = 1
\end{equation}
for almost every $\omega$.
\end{Lemma}

\begin{proof}
We shall show that $\langle N \rangle \to 1$ and $\langle N^2 \rangle \to 1$ in the limit under consideration, which implies the result.

With $\chi_m(\omega)$ denoting the characteristic function of the set corresponding to having exactly $m$ scatterers inside the interval $(0,\nu)$, we have, by symmetry,
\begin{equation}\label{eq:N}
\langle N \rangle = \frac 1 \gamma \sum_{m\geq 2} (m-1) \langle \chi_m \ell_1^{-1} \bar n(\mu\ell_1^2,\infty) \rangle\,.
\end{equation}
To obtain a lower bound, we can restrict the sum to $m\geq 1+(1-\epsilon)\nu$, for some $0<\epsilon<1$. We thus  obtain
\begin{align}\nonumber
\langle N \rangle & \geq \frac{(1-\epsilon) \nu}\gamma \sum_{m\geq 1+(1-\epsilon)\nu} \langle \chi_m \ell_1^{-1} 
\bar n(\mu\ell_1^2,\infty) \rangle \\ & = (1-\epsilon) \left( 1 - \frac \nu 
\gamma \sum_{m < 1+(1-\epsilon)\nu} \langle \chi_m \ell_1^{-1} \bar n(\mu\ell_1^2,\infty) \rangle \right)\,,
\end{align}
where we used (\ref{choice:mu}) the obtain the last equality. From (\ref{bounds:n}) it follows that $\bar n(\mu\ell_1^2,\infty) \leq \mu\ell_1^2$, and hence
\begin{equation}
\langle N \rangle \geq   \left(1-\epsilon\right) \left( 1 - \frac{\nu  \mu}\gamma  \langle P \ell_1 \rangle \right) \quad , \quad P(\omega) =  \sum_{m < 1+(1-\epsilon)\nu}  \chi_m(\omega)\,.
\end{equation}
A Schwarz inequality yields
\begin{equation}\label{schw}
\langle P \ell_1 \rangle \leq \langle P \rangle^{1/2} \langle \ell_1^2 \rangle^{1/2} = \langle P \rangle^{1/2} \frac {\sqrt 2}{\nu}\,,
\end{equation}
where we have evaluated the expectation value of $\ell_1^2$ using (\ref{dpl}).
Moreover, it follows from (\ref{p:l}) that $\langle P \rangle$ is
bounded by a factor that decays exponentially in $\nu$. In particular,
using (\ref{rel-1}), it follows that
\begin{equation}
\liminf \langle N \rangle \geq 1 -\epsilon \,.
\end{equation}
Since $\epsilon$ was arbitrary, this proves the desired lower bound.

For an upper bound, we can proceed similarly. From (\ref{eq:N}), (\ref{choice:mu}) and $\bar n(\mu\ell_1^2,\infty)\leq  \mu\ell_1^2$ we have
\begin{equation}
\langle N \rangle \leq 1+\epsilon  + \frac \mu \gamma  \sum_{m > 1+(1+\epsilon)\nu} (m-1) \langle \chi_m \ell_1\rangle
\end{equation}
for $\epsilon>0$. An similar analysis as above, using the Schwarz inequality as in (\ref{schw}), and (\ref{p:u}) instead of (\ref{p:l}),  then yields
\begin{equation}
\limsup \langle N \rangle \leq 1 +\epsilon \,.
\end{equation}
This shows that $\langle N \rangle \to 1$, as claimed.

Next we consider $\langle N^2\rangle$. Similarly to (\ref{eq:N}), we can write it as
 \begin{align}\nonumber
\langle N^2 \rangle & = \frac 1{\gamma^2} \sum_{m\geq 3} (m-1)(m-2)  \langle \chi_m \ell_1^{-1} 
\bar n(\mu\ell_1^2,\infty) \ell_2^{-1} \bar n(\mu\ell_2^2,\infty) \rangle \\ & \quad  +  
\frac 1 {\gamma^2} \sum_{m\geq 2} (m-1)  \langle \chi_m \ell_1^{-2} \bar n(\mu\ell_1^2,\infty)^2  \rangle \,.\label{eq:N2}
\end{align}
An analysis as above shows that the first sum can be restricted to an
interval $m\in [(1-\epsilon)\nu, (1+\epsilon) \nu]$, and therefore
yields $1$ in the limit considered. Hence it remains to show that last
term in (\ref{eq:N2}) goes to zero. If we restrict the sum to $m-1\leq
2\nu$, say, then we have
\begin{equation}\label{rs2}
\sum_{m =2}^{1+2\nu} (m-1)  \langle \chi_m \bar \ell_1^{-2} n(\mu\ell_1^2,\infty)^2  \rangle \leq 2 \nu \langle 
\bar \ell_1^{-2} n(\mu\ell_1^2,\infty)^2 \rangle \leq 2^4 \frac{\mu^2}{\nu} e^{-\pi \nu /\sqrt{\mu}}
\end{equation}
where we used the upper bound in (\ref{bounds:n}) as well as the bound
\begin{equation}\label{theb}
e^x \int_x^\infty e^{-t}\left(t-\frac {x^2}{t}\right)^k dt \leq k! \,2^k
\end{equation}
for $k=2$. It follows from (\ref{rel-1}) that (\ref{rs2}), when
multiplied by $\gamma^{-2}$, goes to zero as $\nu\to \infty$ if
$\lambda \gg \nu/(\ln \nu)^2$.

For the remaining terms corresponding to $m>1+2\nu$, we use again the Schwarz inequality, which implies that
\begin{equation}\label{45}
\sum_{m> 1+2\nu} (m-1)  \langle \chi_m \bar \ell_1^{-2} n(\mu\ell_1^2,\infty)^2  \rangle  \leq \nu\, \sqrt{2}\, e^{(1-\ln 4) \nu} \langle \ell_1^{-4} \bar n(\mu\ell_1^2,\infty)^4  \rangle^{1/2}\,.
\end{equation}
Because of the upper bound in (\ref{bounds:n}) and (\ref{theb}) for $k=4$ the latter expectation value is above bounded by
\begin{equation}
\langle \ell_1^{-4} \bar n(\mu\ell_1^2,\infty)^4  \rangle \leq 4!\, 2^4 \left( \frac{\mu}{ \nu} \right)^4 e^{-\pi \nu/\sqrt{\mu}} \,.
\end{equation}
With the aid of (\ref{rel-1}) one now readily checks that the right
side of (\ref{45}), when divided by $\gamma^2$, goes to zero as
$\nu\to \infty$. This completes the proof.
\end{proof}

We can take $n_j = (\ell_j\gamma)^{-1} \bar n(\mu\ell_j^2,\infty)$ for
$1\leq j \leq m-1$ as above, and divide each $n_j$ by
$N=\sum_{j=1}^{m-1} n_j$ in order to have the right particle
number. Then (\ref{upper}) implies the upper bound
\begin{equation}\label{upper2}
e_\omega(\gamma,\sigma,\nu) \leq  E \equiv \sum_{j=1}^{m-1} \frac{n_j}{N\ell_j^2} e(n_j \ell_j \gamma/N,\infty)\ .
\end{equation}
The following lemma concludes the proof of the upper bound in Theorem~\ref{thm:gp}.

\begin{Lemma}\label{lem:3}
If $\nu\to \infty$ and $\gamma\to\infty$ with $\gamma\gg \nu/(\ln \nu)^2$ then
\begin{equation}
\lim \frac E{e_0(\gamma,\nu)} = 1
\end{equation}
for almost every $\omega$.
\end{Lemma}

\begin{proof}
We proceed as in the proof of Lemma~\ref{lem:1}, and show that both $\langle E\rangle/e_0(\gamma,\nu)$ and 
$\langle E^2 \rangle/e_0(\gamma,\nu)^2$ converge to $1$ in the limit considered. Because of the result in Lemma~\ref{lem:1}, 
we can replace $N$ by $1$ without loss of 
generality. Using (\ref{inf-n-ord}) (for $\alpha=\infty$) and (\ref{32c}) we see that
\begin{equation}
\frac{n_j}{\ell_j^2} e(n_j \ell_j \gamma,\infty) \sim \mu n_j \,.
\end{equation}
Since also $\mu \sim e_0(\gamma,\nu)$ we can proceed exactly as in the
proof of Lemma~\ref{lem:1} to conclude the result.
\end{proof}

We thus conclude that, for almost every $\omega$,
\begin{equation}
\limsup \frac{ e_\omega(\gamma,\sigma,\nu)}{e_0(\gamma,\nu)} \leq 1
\end{equation}
in the limit $\nu\to \infty$, $\gamma\to \infty$ with $\gamma\gg \nu/(\ln \nu)^2$.


\subsubsection{Lower Bound}

We start from (\ref{below}).
Recall the definition of $g(\mu,\alpha)$ in (\ref{GCE-one-box}).
We have
\begin{equation}\label{lb21}
e_\omega(\gamma,\sigma,\nu) \geq \inf_{\{n_j\}} \sum_{j=0}^m \frac{n_j}{\ell_j^{2}} e(n_j \ell_j \gamma,\ell_j \sigma) \geq \mu + \sum_{j=0}^m \frac 1{\gamma \ell_j^3} g(\mu\ell_j^2,\ell_j\sigma)
\end{equation}
for any $\mu\in\mathbb{R}$.

Using the lower bound in (\ref{32c}), as well as (\ref{e0a}),  have
\begin{equation}
g(\mu\ell^2,\ell\sigma)\geq -\frac{1}{2}\left[ \mu\ell^2 - e(0,\ell\sigma)\right]_+^2 \geq - \frac{\mu^2 \ell^4}{2} \theta(\ell - \widetilde \ell)\,,
\end{equation}
where
\begin{equation}
\widetilde \ell = - \frac 1{2\sigma} + \sqrt{ \frac 1{4\sigma^2} + \frac C \mu }
\end{equation}
and $\theta$ denotes the Heaviside step function. For the first and last interval, corresponding to $j=0$ and $j=m$, respectively, we shall simply use that
\begin{equation}\label{0m}
\frac 1{\gamma \ell_j^3} g(\mu\ell_j^2,\ell_j\sigma) \geq -\frac{\mu^2 \ell_j}{2\gamma} \,.
\end{equation}
When divided by $\mu \sim e_0(\gamma,\nu)$, this goes to zero almost
surely in the limit considered, since $\ell_j \lesssim \nu^{-1}$ and
$\mu \ll \nu \gamma$.

 Pick some small
$\epsilon >0$, and consider the contribution to the sum in
(\ref{lb21}) coming from intervals with length $\ell_j < (\epsilon
\sigma)^{-1}$, for $1\leq j\leq m-1$. If $\epsilon \sigma \geq 1/ \widetilde \ell$, this
contribution is zero, hence we can assume $\epsilon \sigma <
1/\widetilde \ell $ from now on, which is equivalent to
\begin{equation}
\sigma^2 < \frac {\mu(1+\epsilon)} {C\epsilon^2} \,.
\end{equation}
Because of our assumption (\ref{asu}) on $\sigma$, this means that
\begin{equation}
\frac \nu {1+\ln\left(1+\nu^2/\gamma\right)} \ll \sqrt{\mu} \sim \sqrt{\gamma\, f(\nu^2/\gamma)}\,,
\end{equation}
i.e., $\nu^2 \ll \gamma$.
We have
\begin{equation}
\frac 1{\gamma \ell_j^3} g(\mu\ell_j^2,\ell_j\sigma) \geq  - \frac{\mu^2}{2\gamma\epsilon\sigma} \sim - \frac{\mu}{\nu} 
\frac{ \nu}{\sigma (1+\ln(1+\nu^2/\gamma))}
\end{equation}
in this case, and the last fraction goes to zero in the limit considered.

Finally consider the intervals $1\leq j\leq m-1$ with length $\ell_j\geq (\epsilon
\sigma)^{-1}$. In this case, we can use Lemma~\ref{lem1} to bound
\begin{equation}
g(\mu\ell_j^2,\ell_j\sigma)\geq \left( 1 - C\epsilon^{1/2} \right) g( \mu \ell_j^2 (1-C\epsilon^{1/2})^{-1},\infty)\,.
\end{equation}
An either case, we thus have the lower bound
\begin{equation}\label{eith}
\frac 1{\gamma \ell_j^3} g(\mu\ell_j^2,\ell_j\sigma) \geq  \frac{ \left( 1 - C\epsilon^{1/2} \right)}{\gamma \ell_j^3}  
g( \mu \ell_j^2 (1-C\epsilon^{1/2})^{-1},\infty) - 
\frac{C\mu}{\sigma (1+\ln(1+\nu^2/\gamma))}
\end{equation}
for $1\leq j \leq m-1$, for some $C>0$.

Let $\bar \mu = \mu (1-C\epsilon^{1/2})^{-1}$.
With $\bar n(\mu,\infty)$ defined in Subsection~\ref{ss:aux} as the optimizer in (\ref{GCE-one-box}), we can write
\begin{equation}\label{lbg}
g(\bar\mu\ell^2 ,\infty) = \bar n( \bar\mu\ell^2 ,\infty) e(\bar n( \bar\mu\ell^2 ,\infty),\infty) - \bar\mu\ell^2\bar n(\bar\mu\ell^2,\infty)\,.
\end{equation}
We shall choose $\bar \mu$ as in Subsection~\ref{ss:ub}, i.e., such that
\begin{equation}\label{choice:bmu}
 \nu \int_0^\infty dp_\nu(\ell)\, \frac 1{\ell \gamma} \bar n (\bar\mu \ell^2,\infty) = 1\,.
\end{equation}
Then $\bar \mu \sim e_0(\gamma,\nu)$. Let $n_j = (\ell_j \gamma)^{-1} \bar n (\bar\mu\ell_j^2,\infty)$ and $N = \sum_{j=1}^{m-1} n_j$.
From (\ref{lb21}), (\ref{0m}), (\ref{eith}) and (\ref{lbg}) we have
\begin{align}\nonumber
e_\omega(\gamma,\sigma,\nu)  & \geq   \left(1-C\epsilon^{1/2} \right)  \left( \bar\mu \left( 1- N\right) + \sum_{j=1}^{m-1}  
\frac{n_j}{\ell_j^2} e(n_j \ell_j\gamma ,\infty) \right) \\ & \quad - \frac{C\mu m}{\sigma (1+\ln(1+\nu^2/\gamma))} - \frac{\mu^2 
(\ell_0+\ell_m)}{2\gamma} \,.\label{flb}
\end{align}
We have already shown in Lemma~\ref{lem:1} that $\lim N = 1$ almost surely, and in Lemma~\ref{lem:3} that the ratio of the 
sum in (\ref{flb}) and $e_0(\gamma,\nu)$ converges to $1$ almost surely. 
Moreover, after division by $e_0(\gamma,\nu)\sim\mu$ the terms on the last line of (\ref{flb}) vanish almost surely.
This proves that
\begin{equation}
\liminf \frac{e_\omega(\gamma,\sigma,\nu)}{e_0(\gamma,\nu)} \geq 1-C \epsilon^{1/2}
\end{equation}
for almost every $\omega$.
This holds for all $\epsilon>0$, hence the proof of the theorem is complete.

\section{Conclusions}

The main conclusions that can be drawn from the preceding discussion are as follows:
\smallskip

\noindent 1. BEC in the ground state of the interacting gas in the Gross-Pitaevskii regime can survive even in a strong 
random potential with a high density of  scatterers. As far as BEC is concerned the interacting gas in this regime thus 
behaves in a similar way as an ideal 
gas at zero temperature. The character of the wave function of the condensate, however, is strongly affected by the interaction.
\smallskip

\noindent 2. A random potential may lead to localization of the wave function of the condensate, even though the density of 
obstacles is much less than the particle density. The interparticle interaction counteracts this effect and can lead to 
complete delocalization if the interaction is strong enough.
\smallskip

\noindent 3. In terms of the interaction strength, $\gamma$, and density of scatterers, $\nu$, the transition between 
localization and delocalization occurs in the model considered when $\gamma\sim \nu^2$. In the localized regime $\gamma\ll \nu^2$ 
the chemical potential is $\mu \sim (\nu/\ln(\nu^2/\gamma))^2$. For
$\gamma\lesssim \nu/(\ln \nu)$ the condensate is localized in a small number of random intervals.

\section{Appendix 1: Energy Gaps}

We consider the Schr\"odinger operator $H= -\partial_z^2 + W(z)$ on $L^2([0,1])$, with Dirichlet boundary conditions. 
Assume, for the moment, that $W$ is a bounded and continuous function. Without loss of generality, we may assume $W\geq 0$.  
Let $0<e_0<e_1$ denote the lowest two eigenvalues of $H$.

\begin{Lemma}
Define $\eta>0$ by
\begin{equation}\label{deflambda}
\eta^2 = \pi^2 + 3 \int_0^1 W(z) dz
\end{equation}
Then
\begin{equation}\label{gap}
e_1 - e_0 \geq   \eta  \ln \left( 1 + \pi e^{-2\eta}\right)
\end{equation}
\end{Lemma}

For the proof, we follow closely \cite{kirschsimon}.

\begin{proof}
For general $E>0$, let $u(z,E)$ denote the solution of
\begin{equation}\label{schreq}
- u'' + W u = E u
\end{equation}
with $u(0,E)=0$ and $u'(0,E) = 1$ (the prime denoting $\partial_z$). Pick an $\eta > 0$, and introduce the 
Pr\"ufer variables $r(z,E)$ and $\theta(z,E)$ via
$$
u = r \cos \theta \ , \quad u' = -\eta r \sin \theta
$$
The variable $\theta$ is only determined modulo $2\pi$. We fix it uniquely be requiring it to be continuous and 
$\theta(0,E) = -\pi/2$. We note that $\theta$ is increasing in $E$ for fixed $z$, i.e.,  
$\varphi(z,E) = \partial_E \theta(z,E)\geq 0$. This follows from
$$
\varphi = \dot \theta = \eta \frac{ \dot u u' - u \dot u' }{\eta^2 u^2 + (u')^2}
$$
where the dot stands for $\partial_E$. The function $\chi = \dot u u'
- u \dot u'$ satisfies $\chi ' = u^2$ and is thus positive since it
vanishes at $z=0$.

From (\ref{schreq}) it follows that
\begin{equation}\label{thp}
\theta' = \eta^{-1} \left( E - W\right) \cos^2 \theta + \eta \sin^2 \theta
\end{equation}
In particular, $\theta'>0$ near points $z \in [0,1]$ where $u$ vanishes.
Since $u(z,e_0)$ has no zeroes for $z\in(0,1)$, we conclude that $\theta(1,e_0)=\pi/2$. Moreover, $u(z,e_1)$ has exactly one zero in $(0,1)$, hence $\theta(1,e_1) = 3\pi/2$. In particular,
\begin{equation}
\theta(1,e_1)   - \theta(1,e_0) = \pi
\end{equation}

We have $\varphi(0,E) = 0$ for all $E$. From (\ref{thp}) it follows that
\begin{align*}
\varphi' & = \left( \eta + \eta^{-1} \left( W - E \right) \right) \sin(2\theta) \varphi + \eta^{-1} \cos^2(\theta) \\ & \leq \left| \eta + \eta^{-1} \left( W - E \right) \right| \varphi + \eta^{-1}
\end{align*}
From Gronwall's inequality \cite[Thm. III.1.1]{gronwall}, we thus conclude that
$$
\varphi(z,E) \leq \frac z \eta \exp \left( \int_0^z \left| \eta + \eta^{-1} \left( W(y) - E \right) \right| dy \right)
$$
In particular,
$$
\varphi(1,E) \leq  \frac 1 \eta \exp \left( \eta + \eta^{-1} E + \eta^{-1} \int_0^1 W(y) dy \right)
$$
for $E>0$.

This implies that
\begin{align}\nonumber
\pi & = \theta(1, e_1) - \theta(1, e_0) = \int_{e_0}^{e_1} \varphi(1,E) dE \\ & \leq   \exp \left( \eta + \eta^{-1} \int_0^1 W + \eta^{-1} e_0 \right) \left( \exp\left( \eta^{-1} (e_1-e_0) \right) - 1 \right)
\label{dt}
\end{align}
Using the trial function $\sqrt{2} \sin(\pi z)$, we see that $e_0\leq \pi^2 +2 \int_0^1 W$.
From (\ref{dt}) we thus have
$$
e_1- e_0 \geq \eta \ln \left(1+ \pi  \exp \left( -\eta -\frac 3 \eta \int_0^{1}  W(y) dy -\frac {\pi^2}{\eta}   \right) \right)
$$
We choose $\eta$ in order to make the exponent as large as
possible, which leads to the choice (\ref{deflambda}), and concludes
the proof.
\end{proof}

By a simple approximation argument, the bound (\ref{gap}) extends to our case of interest, where
\beq
W(z) = \sigma \sum_{i=1}^m \delta(z-z_i) + \gamma |\psi_0(z)|^2
\eeq
for $\sigma\geq 0$ and $\gamma\geq 0$.
The result is then
\beq\label{gap1}
e_1-e_0 \geq  \eta \ln \left( 1+  e^{-2\eta}\right) \quad \text{with\ } \eta = \sqrt{\pi^2 + 3 m \sigma + 3 \gamma}\,.
\eeq

\subsubsection{Remark.} While the bound \eqref{gap1} on the gap is uniform in the location of the scatterers, we do not expect 
it to be optimal for large $\eta$. In fact, the gap is presumably smallest when all the $z_i$ are equal to $1/2$. In this case, 
a simple calculation (for $\gamma=0$) shows that the gap is asymptotically Êequal to $32\pi^2/(\sigma m)$ Êfor large $\sigma m$.

\section{Appendix 2: Remarks on the Poisson Random Point Field (RPF) on the Real Axis $\mathbb{R}$.}

\subsection{Remark 1} (Poisson and exponential distributions, see, e.g., \cite{LZ}.)

\smallskip

\noindent The following theorem is a consequence of the hypothesis of an independent uniform
distribution of $m$ point "impurities" on the interval $\Lambda_l = [-l/2,l/2] \subset \mathbb{R}$  (homogeneous
$m$-point binomial RPF) and the limit:
$\Lambda_l \rightarrow \mathbb{R} , m \rightarrow \infty$, for a fixed density
\begin{equation}\label{lambda}
\lambda = \lim_{m, \, l \rightarrow\infty} \ \frac{m}{l} \ .
\end{equation}
\begin{theorem}\label{distributions}
{(a)} In the limit (\ref{lambda}) the finite-volume binomial RPF
of impurities $\left\{X_{l}^\omega\right\}$ converges (\textit{in distribution}) to the Poisson RPF
$\left\{X^\omega\right\}$  with the Poisson distribution
\begin{equation}\label{Poiss1}
\mathbb{P}\{\omega:|X^\omega \cap \Lambda_a| = n\} = \frac{(a \lambda)^n}{n!} \ e^{- \lambda a} \ ,
\end{equation}
of the number of impurities in any interval $\Lambda_a$ of the length $a$.\\
{(b)} The uniform distribution of $m-1$ independent points of impurities split up the box $\Lambda_l$ into a set
of sub-intervals $\left\{I_{j}^{\omega}\right\}_{j=1}^{m}$ with $\cup_{j=1}^{m} I_{j}^{\omega} = \Lambda_l$
and their random lengths $\{|I_{j}^{\omega}| = L_{j}^{\omega}\}_{j=1}^{m}$, $\omega\in\Omega$ have the joint
probability distribution
\begin{eqnarray}\label{n-distribution}
\label{measure-L} dP_{l,m}(L_1,...,L_m) = \frac{(m-1)!}{l^{m-1}} \ \delta(L_1+...+L_m - l) \
\ dL_1 dL_2 \ldots dL_m \ .
\end{eqnarray}
Therefore, the random variables $\{L_{j}^{\omega}\}_{j=1}^{m}$ are \textit{dependent}.\\
{(c)} In the thermodynamic limit: $\lambda = \lim_{m, \, l\rightarrow\infty} \ {m}/{l}$, the intervals
$\left\{L_{j}^{\omega}\right\}_{j\geq 1}$ form an infinite set of \textit{independent} random variables and the distribution
corresponding to (\ref{n-distribution}) converges (\textit{weakly}) to the product-measure distribution ${\sigma}_{\lambda}$
defined by the set of consistent marginals:
\begin{eqnarray}\label{marginals}
d{\sigma}_{\lambda, j_1, \ldots,
j_k}(L_{j_1},\ldots,L_{j_k})=\lambda^k\prod_{s=1}^{k}e^{-\lambda
L_{j_s}}dL_{j_s} \ , \ k \in \mathbb{N} \ .
\end{eqnarray}
\end{theorem}

\bigskip

\subsection{Remark 2} (Log intervals \cite{LZ},\cite{JPZ}.)

\smallskip

\noindent We can estimate the maximal length,
$(\max_{1\leqslant j \leqslant m} \, L_{j}^{\omega})_l$, of an interval in the family $\left\{I_{j}^{\omega}\right\}_{j=1}^{m}$
of random sub-intervals of $\Lambda_l =[-l/2,l/2]$ when $l \rightarrow \infty$.
To this end we define the events
\begin{equation}\label{event}
A_{l}^{\delta}:=\{\omega: \frac{(\max_{1\leqslant j \leqslant m} \, L_{j}^{\omega})_l}{\ln l} \,
\leqslant \, \frac{\delta}{\lambda}\} \ , \  l \in \mathbb{N} \  {\rm{and}} \ l > 2 \ ,
\end{equation}
as well as the following ``tail" event:
\begin{equation}\label{tail-event}
A_{\infty}^{\delta} := \limsup_{l \rightarrow \infty} A_{l}^{\delta} =
\bigcap_{n=1}^{\infty} \bigcup_{l=n}^{\infty} A_{l}^{\delta} \ .
\end{equation}
By definition, (\ref{tail-event}) is the event that for increasing $l$ the length of the largest subinterval in $\Lambda_l$ is bounded from
above by the nonrandom number $(\delta \ln l)/\lambda$ for all but finitely many $l\in \mathbb{N}$. The following theorem holds:

\begin{theorem}\label{theorem-log}
Let $\lambda > 0$ be a mean concentration of the point Poisson ``impurities" on $\mathbb{R}$.
Then \\
(i) for any $\delta > 4$, one has:
\begin{equation}\label{> 4}
\mathbb{P}(A_{\infty}^{\delta}) = 1 \ ,
\end{equation}
(ii) for any $\delta \leq 2$, one has:
\begin{equation}\label{< 2}
\mathbb{P}(A_{\infty}^{\delta}) = 0 \ .
\end{equation}

\end{theorem}
\noindent \textbf{Proof}: (i) The complementary event  to (\ref{event}) is
\begin{eqnarray*}
(A_{l}^{\delta})^c := \big\{ \omega:  (\max_{1\leqslant j \leqslant m} \, L_{j}^{\omega})_l \, > \,
\frac{\delta}{\lambda} \, \ln l  \big\}.
\end{eqnarray*}
Let $N_l:= \big[{\lambda l}/{(\delta \ln l)}\big] +1$, where $[x]$ denotes the integer part of $x\geq0$, and define a
new interval:
\begin{equation}\label{new-interval}
{\widetilde{\Lambda}}_{l} := [-{N_l}(\frac{\delta}{2\lambda} \ln l)\, , \,
{N_l}(\frac{\delta}{2\lambda} \ln l)] \, \supset \, \Lambda_{l} \ .
\end{equation}
Split this bigger interval into $2 N_l$ identical disjoint intervals $\{J_{i}^{l}\}_{i = 1}^{2 N_l}$ of size
${\delta}{(2\lambda)^{-1}} \ln l$. If the event $(A_{l}^{\delta})^c$ occurs, then there
exists at least one empty interval $J_{i}^{l}$ (interval without any impurities), i.e.,  
\begin{eqnarray}\label{event-empty}
(A_{l}^{\delta})^c \, \subset \, \bigcup_{1 \leqslant i \leqslant 2 N_l} \{\omega: J_{i}^{l} \ \textrm{is empty} \} \ .
\end{eqnarray}
Note that for a given number of $m-1$ impurities on the finite interval $\Lambda_l$ the RPF 
$\left\{X_{l}^\omega\right\}$ is \textit{binomial} for any two sets: $J_{i}^{l}$ and $\Lambda_{l}\setminus J_{i}^{l}$.
Since the probability for the interval $J_{i}^{l}$ to be empty depends only on its size 
${\delta}{(2\lambda)^{-1}} \ln l$ and the size of $\Lambda_l$, one obtains by (\ref{new-interval}) and (\ref{event-empty})
the estimate
\begin{equation}\label{binom-Pr-empty}
\mathbb{P}^{\, bin}_{m-1,l}((A_{l}^{\delta})^c) \,
\leq 2 N_l\left(1- \frac{\delta}{2\lambda} \ln l/|\widetilde{\Lambda}_{l}|\right)^{m-1} \leq
2 N_l\left(1- \frac{1}{2 N_l}\right)^{(N_l -1)\, \delta \ln l} \ .
\end{equation}
To get the last inequality we put $m-1 = [\lambda \, l] + 1 \geq \lambda \, l$, see (\ref{lambda}), 
and we use that $(N_l - 1)\, \delta \ln l \leq \lambda \, l$. Since the estimate (\ref{binom-Pr-empty}) yields
\begin{equation}\label{lim-binom}
\lim_{l \rightarrow \infty}
\mathbb{P}^{\, bin}_{[\lambda \, l]+1,l}((A_{l}^{\delta})^c)/\left[2 ([{\lambda l}/{(\delta \ln l)}] +1) 
\ l^{-\delta/2}\right] \leq 1 \ ,
\end{equation}
we obtain that for $\delta > 4$
\begin{eqnarray}\label{sum-bin}
\sum_{l \geqslant 1} \mathbb{P}^{\, bin}_{[\lambda \, l]+1,l}((A_{l}^{\delta})^c) \, < \, \infty \ .
\end{eqnarray}
By Theorem \ref{distributions} binomial distribution $\mathbb{P}^{\, bin}_{[\lambda \, l]+1,l}(\cdot)$ converges
to the Poisson distribution $\mathbb{P}(\cdot)$ on $\mathbb{R}$, when $l \rightarrow \infty$. Then
estimate (\ref{lim-binom}) and (\ref{sum-bin}) imply
\begin{eqnarray*}
\sum_{l \geqslant 1} \mathbb{P}((A_{l}^{\delta})^c) \, < \, \infty.
\end{eqnarray*}
Therefore, by the Borel-Cantelli lemma $\mathbb{P}((A_{\infty}^{\delta})^c)=0$, which is equivalent to (\ref{> 4}).
This also implies that there exists a subset $\widetilde{\Omega}\subset\Omega$
of full measure, $\mathbb{P}(\widetilde{\Omega}) =1$, such that for each $\omega\in \widetilde{\Omega}$
one can find $l_{0}(\omega) < \infty$ with
\begin{eqnarray*}
\mathbb{P} \ \{\omega: (\max_{1\leqslant j \leqslant m} \, L_{j}^{\omega})_l \, \leqslant \,
\frac{\delta}{\lambda} \, \ln l\} = 1 \ , \ {\rm{for}} \ \delta > 4 \ ,
\end{eqnarray*}
for all $l \geqslant l_{0}(\omega)$ and $m \geq [\lambda \, l]$. \\
(ii) Consider the event $(A_{l}^{\delta})^c $ and note that one can estimate its binomial probability
from below taking only one term in the union (\ref{event-empty}). Then
\begin{equation}\label{binom-Pr-emptyBIS}
\mathbb{P}^{\, bin}_{m-1,l}((A_{l}^{\delta})^c) \,
\geq \left(1- \frac{\delta}{2\lambda} \ln l/|\widetilde{\Lambda}_{l}|\right)^{m-1} \geq
\left(1- \frac{\delta}{2\lambda} \ln l/[\lambda \, l]\right)^{[\lambda \, l]} \ ,
\end{equation}
where we put $m-1 = [\lambda \, l]$, see (\ref{lambda}), and we use (\ref{new-interval}). Then (\ref{binom-Pr-emptyBIS})
yields
\begin{equation}\label{lim-binomBIS}
\lim_{l \rightarrow \infty}
\mathbb{P}^{\, bin}_{[\lambda \, l],l}((A_{l}^{\delta})^c)/ (l^{-\delta/2}) \geq 1 \ ,
\end{equation}
By the same arguments as in (i) one concludes from (\ref{lim-binomBIS}) that for $\delta \leq 2$
\begin{eqnarray*}
\sum_{l \geqslant 1} \mathbb{P}((A_{l}^{\delta})^c) \, = \, \infty.
\end{eqnarray*}
Since the events $\{(A_{l}^{\delta})^c\}_{l > 2}$ are independent, by the Borel-Cantelli lemma we obtain:
$\mathbb{P}((A_{\infty}^{\delta})^c)=1$, or equivalently (\ref{< 2}). This again can be translated as
\begin{eqnarray*}
\mathbb{P} \ \{\omega: (\max_{1\leqslant j \leqslant m} \, L_{j}^{\omega})_l \, \geq \,
\frac{\delta}{\lambda} \, \ln l\} = 1 \ , \ {\rm{for}} \ \delta \leq 2 \ ,
\end{eqnarray*}
for all $l \geqslant l_{0}(\omega)$ and $m > [\lambda \, l]$.    \hfill $\square$

\bigskip

Since the Poisson RPF on $\mathbb{R}$ conditional to the interval $[-l/2,l/2]$ and to $m-1$ impurities
\textit{coincides} with the binomial RPF, instead of increasing boxes with finite binomial RPF configurations
one can consider \textit{restrictions} of infinite Poisson RPF configurations to the increasing set of windows $[-l/2,l/2]$, 
cf \cite{JPZ}.
Then Theorem \ref{theorem-log} claims that \textit{almost surely} (with respect to the Poisson RPF distribution) 
the maximal length of impurity-free sub-intervals, $(\max_{1\leqslant j \leqslant m} \, L_{j}^{\omega})_l$, in this 
windows is bounded from below and above by $({\rm const }) \ln l/\lambda$.

\bigskip


 \end{document}